\documentclass[11pt]{iopart}
\expandafter\let\csname equation*\endcsname\relax
\expandafter\let\csname endequation*\endcsname\relax
\usepackage{color}      
\usepackage{amsthm,amssymb,amsfonts,bbm,mathtools,mathrsfs}
\usepackage{graphicx,thmtools,hyperref,times,subfigure,}

\usepackage{framed}

\usepackage{braket}			
\renewcommand{\d}[2]{\frac{d #1}{d #2}} 

\renewcommand{\tr}{{\rm Tr}}
\newcommand{\diag}{{\rm diag}}
\newcommand{\eig}{{\rm eig}}

\newtheorem{lemma}{Lemma}
\newtheorem{proposition}{Proposition}
\newtheorem{theorem}{Theorem}

\newcommand{\proofend}{\hfill\fbox\\\medskip }
\renewenvironment{proof}{{\noindent \bf Proof }}{$\proofend$} 

\declaretheoremstyle[notefont=\bfseries,notebraces={}{},%
    headpunct={.},postheadspace=0.5em,bodyfont=\it]{mystyle}
\declaretheorem[style=mystyle,numbered=no,name=Theorem]{thm-hand}

\graphicspath{{figures/}} 

\begin{document}
\title{Locally optimal symplectic control of multimode Gaussian states}
\author{Uther Shackerley-Bennett$^1$, Alberto Carlini$^{2,3,4}$, Vittorio Giovannetti$^{4,5}$, Alessio Serafini$^1$}
\address{$^1$Department of Physics \& Astronomy, University College London, London WC1E 6BT, UK \\ 
$^2$Dipartimento di Scienze ed Innovazione Tecnologica, Universit\`a del Piemonte Orientale, 15121 Alessandria, Italy \\
$^3$INFN, Sezione di Torino, Gruppo Collegato di Alessandria, Italy \\
$^4$NEST, Istituto di Nanoscienze-CNR, Pisa, Italy \\
$^5$Scuola Normale Superiore, 56127 Pisa, Italy}
\ead{uther.shackerley-bennett.14@ucl.ac.uk}
\begin{abstract}
The relaxation of a system to a steady state is a central point of interest in many attempts to advance control over the quantum world. 
In this paper, we consider control through instantaneous Gaussian unitary operations on the ubiquitous lossy channel, 
and find locally optimal conditions for the cooling and heating of a multimode Gaussian state subject to losses
and possibly thermal noise. This is done by isolating the parameters that encode entropy and temperature and by deriving an equation for their evolution. This equation is in such a form that it grants clear insight into how relaxation may be helped by instantaneous quantum control. It is thus shown that squeezing is a crucial element in optimising the rate of change of entropic properties under these channels. Exact relaxation times for heating and cooling are derived, up to an arbitrarily small distance from the fixed point of the lossy channel with locally optimal strategies.
\noindent{\it Keywords\/}: {Gaussian states, coherent control, open quantum systems}
\end{abstract}

\maketitle

\section{Introduction}\label{sec:intro}
As quantum technologies advance, it has become increasingly important to understand the dynamics of a system interacting with its environment. It is hoped that increased insight into such setups will provide a more realistic notion of how to control and exploit quantum behaviours. As well as considering the impact of the environment, one would also like to allow for an external controller who may influence the evolution of the state of the system. This beckons the introduction of open-loop quantum control theory, which is a mathematical framework exploring how time-varying Hamiltonians drive dynamics in the absence of feedback. Exploring these together in a particular scenario is the topic of this study.

Quantum control theory of closed systems is a highly active area of research \cite{Wu2007, Wu2008, Genoni2012, Shack2016, Dalessandro2007} with its roots in the older field of mathematical control theory \cite{JurdjevicGeo, Jurdjevic1972, Sachkov2000, Elliott2009}. Its introduction to open systems has taken place in different contexts, including multilevel discrete systems \cite{Jirari2005}, systems with closed `feedback' control \cite{Lloyd2001}, dissipating qubits \cite{Mukherjee2013} and in a drive to understand `quantum speed limits' \cite{Campo2013}. The current paper explores open system control theory for continuous variable quantum mechanics, and specifically for the class of Gaussian states. These are ubiquitous in physics and serve as a good model for electromagnetic radiation \cite{Bucca,Weedbrook2012, Ferraro2005}, optomechanical systems \cite{Aspelmeyer2014}, trapped ions \cite{Leibfried2003} and mesoscopic massive systems seeking the `gravitational quantum regime' \cite{Pino2016, Romero2011}. They are fundamental in the study of continuous variable quantum information due to the ease with which they can be described, and also the natural way in which noise may be introduced into the dynamical equations. The evolution of Gaussian states in dissipative quantum channels has been explored before and it is this literature that we build upon \cite{Carlini2014, Genoni2016}. 

If one models the environment as an infinite thermal bath coupled to the system through a beam-splitter Hamiltonian, then the channel is known as \textit{lossy} \cite{Eisert2005}. Lossy evolutions can be described as either `heating' or `cooling', depending on whether the environment has a higher or lower entropy than the initial state of the system. 
Here we consider the case of $n$ non-interacting bosonic modes, each coupled with a bosonic thermal environment with the same temperature, undergoing such a lossy evolution. We then ask for the locally optimal, open-loop control strategy to minimise the relaxation time to the steady state. In other words, we wish to maximise the rate of change of the entropy toward the fixed point of the dynamics. 
Control is considered in the idealised regime where only instantaneous transformations which act impulsively on the system are allowed, in particular we restrict the analysis to (possibly non-local) unitaries that are Gaussian preserving.
The restriction to finite, instantaneous controls (corresponding to control Hamiltonians proportional to delta functions in time)
is not a completely wild abstraction since 
experimental set-ups certainly exist where a 
unitary manipulation will typically take nanoseconds, or tens of nanoseconds, while 
the decoherence rates are in the order of $10-10^3$ ${\rm kHz}$, 
so that the former may fairly be regarded as instantaneous with respect to the time-scales 
on which the noise acts. 
In what follows, we also disregard the possibility of any feedback on the quantum system.
Let us anticipate that, in the absence of specific constraints, 
the non-compact nature of the symplectic group allows for an arbitrarily high amount of squeezing to occur in these control operations. To make the analysis more realistic and reduce infinities, this will be capped.

A central result of this study is showing that 
squeezing is the main parameter of interest in the analysis of heating and cooling dynamics. A highly squeezed state is locally optimal for minimising the relaxation time in a heating channel. Interestingly, squeezing can also be used to drive the state away from the fixed point of a cooling channel. Conversely, undoing the squeezing is shown to be locally optimal for cooling strategies, which is a point of experimental interest. The minimum time it takes to reach the fixed point of the channel within an arbitrarily small distance is derived in both cases. The analysis rests on the derivation of 
a compact equation for the evolution of entropic quantities -- Eq.~(\ref{eq:masterm}) -- whose variations might also be applicable to the evolution of entanglement of Gaussian states under lossy channels, as discussed in the conclusion.

\section{Gaussian states}\label{sec:gauss}
We begin with an outline of Gaussian states, referring to the introductions given in Refs.~\cite{Bucca,Olivares2011}.

Let $\hat{\mathbf{r}} = (\hat{x}_1,\hat{p}_1,\ldots, \hat{x}_n,\hat{p}_n)^\intercal$ be a vector of canonical operators such that $[\hat{x}_j,\hat{p}_j] = i\delta_{jk}$, where $\delta_{jk}$ is the Kronecker delta function. From this we find that $[\hat{\mathbf{r}},\hat{\mathbf{r}}^\intercal] = i\Omega$ where
\begin{equation}
\Omega = \bigoplus_{i=1}^n \begin{pmatrix} 0 & 1 \\ -1 & 0 \end{pmatrix}.
\end{equation}
Quadratic Hamiltonians are defined as those that can be written as $\hat{H} = \frac{1}{2}\hat{\mathbf{r}}^\intercal H \hat{\mathbf{r}} + \hat{\mathbf{r}}^\intercal \mathbf{a}$ where $H$ is a $2n\times 2n$, real, symmetric matrix and $\mathbf{a}$ is a vector of real numbers. The set of Gaussian states can be defined as the ground and thermal states of positive definite quadratic Hamiltonians:
\begin{equation}\label{eq:system}
\hat{\rho}_G = \frac{e^{-\beta \hat{H}}}{\tr[e^{-\beta \hat{H}}]},
\end{equation}
where $\beta $ is the inverse temperature of the state (we set Boltzmann's constant $k_B=1$). Such states are referred to as Gaussian due to their Wigner representation which takes a Gaussian form. The Wigner representation immediately suggests that such states should be totally specified by their first and second moments, defined respectively as
\begin{equation}
\mathbf{d} = \tr[\hat{\mathbf{r}}\hat{\rho}_G],
\end{equation}
\begin{equation}
\sigma = \tr[\{(\hat{\mathbf{r}}-\mathbf{d}),(\hat{\mathbf{r}}-\mathbf{d})^\intercal\}\hat{\rho}_G].
\end{equation}
$\sigma$ is the $2n \times 2n$ covariance matrix of the Gaussian state and $\{\cdot,\cdot\}$ denotes the anticommutator. 

Evolution under quadratic Hamiltonians transforms Gaussian states into Gaussian states. The second moment encodes all entropic and entanglement properties and so we restrict our attention, here and in the rest of the paper, to ignore the first moments. In this case it is possible to consider $\mathbf{d}=0$ in both the state and the quadratic Hamiltonian transformation. Once this is done the transformations generated by the remaining Hamiltonians, and their combinations, form the symplectic group $\operatorname{Sp}(2n,\mathbb{R})$. This is the set of $2n \times 2n$ matrices $S$ such that
\begin{equation}
S\Omega S^\intercal = \Omega.
\end{equation}
Covariance matrices transform under a finite dimensional representation of the symplectic group by conjugation,
\begin{equation}
\sigma \to S\sigma S^\intercal.
\end{equation}

\section{Open Diffusive dynamics}\label{sec:diff}
The details of Markovian evolution for Gaussian states are presented in Refs.~\cite{Bucca,Genoni2016} where, following a standard approach in quantum optics, the weak coupling to an infinite bath is replaced with the coupling to a finite number of 
continuously refreshed environmental modes. 
This coupling is provided by the Hamiltonian,
\begin{equation}
 H_C = \begin{pmatrix} 0 & C \\ C^\intercal & 0 \end{pmatrix}.
\end{equation}
The evolution of the Gaussian state has a representation in terms of the associated covariance matrix which takes the form
\begin{equation}\label{eq:covmastergen}
\dot{\sigma} = A\sigma + \sigma A^\intercal + D,
\end{equation}
where 
\begin{equation}
 A = \frac{1}{2}\Omega C \Omega C^\intercal \quad \text{and} \quad D = \Omega C \sigma_B C^\intercal \Omega^\intercal,
\end{equation}
with $\sigma_B$ denoting the bath covariance matrix.

Lossy channels occur by setting the coupling Hamiltonian to induce a swapping of field excitations between the system and the bath. 
We focus on the case of $n$ non-interacting bosonic modes, each coupled with a bosonic thermal environment via an individual exchange Hamiltonian of the form:
\begin{equation}
\hat{H}_C = \sqrt{\eta} (\hat{a}\hat{b}^\dagger + \hat{a}^\dagger \hat{b}),
\end{equation}
where $\hat{a}$ denotes the annihilation operator of the system and $\hat{b}$ the annihilation operator of the bath mode. This is equivalent to setting $C = \sqrt{\eta} \Omega$ which in turn sets $A = -\eta\mathbb{I}/2$ and $D = \eta \sigma_B$. We set the state of the bath as $\sigma_B = \chi \mathbb{I}$ where $\chi = 2\bar{N}+1$ and $\bar{N}$ is the average photon number in each bath mode. 
We assume that the $n$ thermal baths have the same $\bar N$, i.e. the same temperature. Rescaling the time parameter by $\eta$, the relaxation rate of the lossy channel, we arrive at
\begin{equation}\label{eq:master}
\dot{\sigma} = -\sigma + \chi \mathbb{I},
\end{equation}
which represents the free, lossy evolution of the Gaussian state with solution
\begin{equation}\label{eq:mastersol}
\sigma = \chi \mathbb{I} + (\sigma(0) - \chi\mathbb{I})e^{-t}.
\end{equation}
Eq.~(\ref{eq:master}) will be used extensively in the rest of the paper as we explore the evolution of some important parameters of $\sigma$.

\section{Symplectic invariants}\label{sec:sympinv}
An $n$-mode covariance matrix has $n(2n+1)$ free parameters that compose its information content. $n$ of these parameters have the special status of being invariant under symplectic transformations and so play a unique role in describing the state. These symplectic invariants are important in that they encode the entropy of the state, relating to its temperature and mode frequency \cite{Serafini2006, Serafini2004}. We will explore two different ways of representing these invariants allowing the derivation of a new equation from Eq.~(\ref{eq:master}) that isolates their evolution. This analysis will clear the route to an understanding of entropy evolution and will provide a platform on which to introduce control.

There is no unique way to define the symplectic invariants but the collection most commonly considered is the set of symplectic eigenvalues, expressible through Williamson's theorem \cite{Olivares2011, Williamson1936}. This states that any $2n$-dimensional, real, symmetric, positive-definite matrix, for instance $\sigma$, may be decomposed as 
\begin{equation}
 \sigma = SWS^\intercal,
\end{equation}
where
\begin{equation}
W =\bigoplus_{i=1}^n\nu_i \mathbb{I}_2,
\end{equation}
and $\mathbb{I}_2$ is the $2\times 2$ identity matrix, 
$W>0$ and $S \in \operatorname{Sp}(2n,\mathbb{R})$. The elements of $W$ are unique up to reordering. Furthermore the uncertainty principle on $\sigma$ ensures that $\nu_i \geq 1$ \cite{Olivares2011}. The parameters $\nu_i$ are the symplectic eigenvalues which can also be calculated as moduli of the eigenvalues of $\Omega\sigma$.

Any function of symplectic invariants is also a symplectic invariant. Consider the $k$th elementary symmetric function of the eigenvalues $\lambda_i$ of some $m \times m$ matrix $X$, defined as \cite{Meyer2000}
\begin{equation}
\vartheta_k[X] := \sum_{E \in \mathcal{E}_k^m} \prod_{j \in E} \lambda_j,
\end{equation}
where the sum runs over all the possible $k$-subsets $E \in \mathcal{E}_k^m$ of the first $m$ natural integers. More explicitly,
\begin{equation}
E \in \mathcal{E}_k^m \subset P(\mathbb{N}_m) \quad \text{iff} \quad |E| = k,
\end{equation}
where $\mathbb{N}_m = \{1,\ldots,m\}$, $P(\cdot)$ denotes the power set and we put $\vartheta_0[X] := 1$. As an example, for a matrix $X$ with eigenvalues $\{\lambda_1,\lambda_2,\lambda_3,\lambda_4\}$, 
\begin{equation}
\vartheta_3[X] = \lambda_1\lambda_2\lambda_3 + \lambda_1\lambda_3\lambda_4 + \lambda_1\lambda_2\lambda_4 + \lambda_2\lambda_3\lambda_4.
\end{equation}

Following Ref.~\cite{Serafini2007a} this allows for the definition of a new set of symplectic invariants of $\sigma$, given by $\vartheta_{2k}[\Omega\sigma]$ and related to the symplectic eigenvalues via
\begin{equation}
\vartheta_{2k}[\Omega\sigma] = \sum_{E \in \mathcal{E}_k^n}\prod_{j \in E} \nu_j^2.
\end{equation}
The elementary symmetric functions of odd order vanish in this case because of the alternating signs of the eigenvalues of $\Omega \sigma$. Both the set of symplectic eigenvalues and this new set contain the entropic information of the state.

To further our analysis of these new invariants we recall some basic linear algebra \cite{Meyer2000}. Consider an $m\times m$ matrix $X$ and delete the same $m-r$ rows and columns. The remaining $r \times r$ submatrix is known as a principal submatrix of $X$. The determinant of a principal submatrix is known as a principal minor. Given the characteristic polynomial $\sum_{k=0}^m c_k \lambda^{m-k}$ of $X$, where $c_0 = 1$, we have that
\begin{equation}
\begin{aligned}
c_k[X] &= (-1)^k \sum \text{ (all $k\times k$ principal minors)}, \\
\vartheta_k[X] &= \sum \text{ (all $k \times k$ principal minors)}.
\end{aligned}
\end{equation}
Noting that $c_{2k}[X] \equiv \vartheta_{2k}[X]$ we see that the study of principal minors and characteristic polynomial coefficients gives a new way to express and understand the symplectic invariants of a covariance matrix. In fact there is a way to recursively generate $c_k[X]$ using Fadeev-Le Verrier recursion.
\begin{theorem}[Fadeev-Le Verrier recursion \cite{Helmberg1993, Fadeev1963, Hou1998}]
 Let $X$ be an $m \times m$ real matrix. Let its characteristic polynomial be written $\det (X-\lambda \mathbb{I}) = \sum_{k=0}^m c_k \lambda^{m-k}$ with $c_0 = 1$. It is possible to calculate the coefficients of the polynomial via the recursive formula,
\begin{equation}\label{eq:lev}
c_k[X] = -\frac{1}{k} \sum_{i=0}^{k-1} \tr [X^{k-i}] c_i[X].
\end{equation}
\end{theorem}

This recursive generation of $c_k[X]$ will be the guide to stripping Eq.~(\ref{eq:master}) to only consider the evolution of the invariants. Before stating this evolution equation it is necessary to spend a little more time on notation, defining 
 \begin{equation}
\bar{\vartheta}_{2k}^i[\Omega\sigma] := \sum_{E \in \mathcal{E}_{k}^{n, i}} \prod_{j \in E} \nu_j^2,
 \end{equation}
where 
\begin{equation}
E \in \mathcal{E}_k^{n,i} \subset P(\mathbb{N}_n\setminus \{i\}) \quad \text{iff} \quad |E| = k.
\end{equation}
This new object acts as a sort of reduced $\vartheta_{2k}[\Omega\sigma]$ where we remove the terms involving the $i$th symplectic eigenvalue. Again we define $\bar{\vartheta}_{0}^i[\Omega\sigma] := 1$. To illustrate this new function consider some covariance matrix $\sigma$ with symplectic eigenvalues $\{\nu_1,\nu_2,\nu_3,\nu_4\}$. Here, we would have
\begin{equation}
\bar{\vartheta}^2_4[\Omega\sigma] =  \nu_1^2\nu_3^2 + \nu_1^2\nu_4^2 + \nu_3^2\nu_4^2.
\end{equation}

Using Fadeev-Le Verrier recursion to express the characteristic polynomial coefficients and their link to symplectic eigenvalues we are able to derive an evolution equation for these invariants under the open diffusive dynamical equation, Eq.~(\ref{eq:master}).
\begin{theorem}\label{thm:masterm}
For $\sigma$ evolving under $\dot{\sigma} = -\sigma + \chi\mathbb{I}$, the evolution of the symplectic invariants, defined by $\vartheta_{2k}[\Omega\sigma]$, obeys
\begin{equation}\label{eq:masterm}
\dot{\vartheta}_{2k}[\Omega\sigma] = -2k\vartheta_{2k}[\Omega\sigma] +\chi \tr[SV_{k}S^\intercal],
\end{equation}
where $\sigma = SWS^\intercal$ by the Williamson decomposition and
\begin{equation}
V_{k} = \bigoplus_{i=1}^n \left(\nu_i \bar{\vartheta}_{2(k-1)}^i[\Omega\sigma]\right)\mathbb{I}_2.
\end{equation}
\end{theorem}
\begin{proof}
See~\ref{sec:sympinvevol}.
\end{proof}

Eq.~(\ref{eq:masterm}) directly reveals the evolution of the $n$ invariants under Eq.~(\ref{eq:master}), meaning that we are able to focus on the parameters that interest us and ignoring the $2n^2$ remaining ones.

Open loop control refers to the enactment of predetermined operations without feedback on a quantum state, which would typically be unitary. The importance of the symplectic group is that it is the Gaussian analogue of the unitary group, in that the elements do not alter the entropic, or Gaussian, properties of the state. We consider symplectic control acting instantaneously at a given time in the evolution, allowing the alteration of $S$ on the right hand side of Eq.~(\ref{eq:masterm}). 
Notice that this represents an impulsive action altering the state by a finite amount at a certain time and not a term generated by a 
Hamiltonian proportional to a step function in time (the $S$ in Eq.~(\ref{eq:masterm}) is a time-dependent property of the evolving state).
This is with the aim of optimising the rate of the change of $\vartheta_{2k}$, where the argument of $\vartheta_{2k}$ is assumed to be $\Omega \sigma$ unless otherwise stated.

Note that $S$ is defined via the Williamson decomposition of $\sigma$ and so both terms on the right hand side of Eq.~(\ref{eq:masterm}) could be affected by its alteration. The first term, however, is a symplectic invariant and so no amount of manipulation will alter its value. However, the second term does vary with $S$ and so if, at a given point in the evolution, we alter it then this term will change. Maximising or minimising the value of the trace term is therefore the route towards locally optimising the rate of change of $\vartheta_{2k}$. In Sec.~\ref{sec:opt} we find that this optimisation will allow for the study of decoupled dynamics where each mode evolves independently. This drastically simplifies the analysis and allows for the optimal relaxation times for heating and cooling. We will show that squeezing is the key parameter that alters this term in Eq.~(\ref{eq:masterm}).

Notice also that all possible initial Gaussian states are allowed in our treatment, since all physical covariance matrices admit a Williamson decomposition in terms of some 
symplectic $S$.

\section{Heating and cooling}\label{sec:opt}
Eq.~(\ref{eq:masterm}) shows that varying $S$ to alter the trace term is the required action to optimise the rate of change of $\vartheta_{2k}$. First we derive the symplectic matrix that will maximise or minimise this term and then we turn our attention to the dynamics that this optimisation invokes. This latter analysis will provide locally optimal heating and cooling times for the evolution of the channel under symplectic control.

\subsection{Trace optimisation}
To explore the optimisation of the trace term in Eq.~(\ref{eq:masterm}) we first require an understanding of the singular value decomposition of symplectic matrices. Any symplectic matrix may be decomposed as $S = R_1ZR_2$ where $R_1,R_2 \in \operatorname{OSp}(2n,\mathbb{R}) = \operatorname{O}(2n) \cap \operatorname{Sp}(2n,\mathbb{R})$ and $Z = \diag(z_1, 1/z_1,\ldots z_n, 1/z_n)$. The elements of $\operatorname{OSp}(2n,\mathbb{R})$ correspond to `passive' elements in the lab such as phase shifters and beam splitters. The $Z$ component encodes the squeezing element of the transformation. Any symplectic that is singular value decomposed such that $Z \neq \mathbb{I}$ is referred to as `active'. Using this decomposition will allow us to deconstruct and understand the trace term of Eq.~(\ref{eq:masterm}).

For generality and brevity of notation we will consider a trace term of the form $\tr[SYS^\intercal]$ where $Y = \bigoplus_{i=1}^n y_i\mathbb{I}_2$, with $y_i$ positive, and $S \in \operatorname{Sp}(2n,\mathbb{R})$. 
Note that the elements of $Z$ are unbounded and closely related to the energy required to 
enact the operation. 
Therefore the full optimisation will be done 
over the set $\overline{\operatorname{Sp}}(2n,\mathbb{R})$, defined as the restriction of the symplectic 
group to elements with maximum singular value $\overline{z}_i$ for each mode. This encapsulates the reasonable assumption that the energy that may be employed in each mode, rather than the sum of the energies, is bounded.
\begin{proposition}\label{thm:tropt}
The supremum of the trace term over the set $\overline{\operatorname{Sp}}(2n,\mathbb{R})$ is given by
\begin{equation}\label{eq:prop1}
 \sup_{S \in \overline{\operatorname{Sp}}(2n,\mathbb{R})}\tr[S Y S^\intercal] = \sum_{i=1}^n 2\zeta^+_{\bar{z}_i} y_i,
\end{equation}
where $\zeta^{\pm}_z := \frac{z^{2}\pm 1/z^{2}}{2}$, $\bar{z}_i$ is the maximum local squeezing allowed on mode $i$, ordered such that $\bar{z}_1 \geq \ldots \bar{z}_n$, and $y_1 \geq \ldots \geq y_n$. The infimum is given by $\tr[Y]$.
\end{proposition}
\begin{proof}
The singular value decomposition provides
\begin{equation}\label{eq:sup1}
\begin{aligned}
\tr[S Y S^\intercal] &=\tr[ R_1 Z R_2 Y R_2^\intercal Z R_1^{\sf T}] \\
&=\tr[Z^2 R_2 Y R_2^\intercal].
\end{aligned}
\end{equation}
It is possible to change to a basis in which the symplectic form becomes $\Omega \to \begin{pmatrix} 0_n &  \mathbb{I}_n \\ -\mathbb{I}_n & 0_n \end{pmatrix}$, where $0_n$ denotes the $n\times n$ zero matrix and $\mathbb{I}_n$ the $n\times n$ identity matrix. This is a commonly used basis in Gaussian state theory and is enacted with
\begin{equation}
 P_{kl} = \begin{cases} 1, \quad k\leq n, l=2k-1, \\ 1, \quad k > n, l=2(k-n) \end{cases}
\end{equation}
acting by similarity on the components of the trace term \cite{Ferraro2005}. After enacting this we will then consider a further similarity transformation with
\begin{equation}
 Q = \frac{1}{\sqrt{2}}\begin{pmatrix} \mathbb{I}_n & i\mathbb{I}_n \\ \mathbb{I}_n & -i\mathbb{I}_n \end{pmatrix}.
\end{equation}
This transforms symplectics into a basis which highlights the isomorphism between $\operatorname{OSp}(2n,\mathbb{R})$ and $\operatorname{U}(n)$, as will become clear. Acting on each matrix in the trace by $P$ and then $Q$ we find that the transformation of $Z^2$ gives
\begin{equation}
 Z^{2 \prime} := (QP)Z^2(QP)^{-1} = \begin{pmatrix} \Gamma^+ & \Gamma^- \\ \Gamma^- & \Gamma^+ \end{pmatrix},
\end{equation}
where $ \Gamma^\pm(z) := \operatorname{diag}(\zeta^{\pm}_{z_1},\ldots, \zeta^{\pm}_{z_n})$. $Y$ transforms as
\begin{equation}
Y' := (QP)Y(QP)^{-1} = \begin{pmatrix} \Upsilon & 0_n \\ 0_n & \Upsilon \end{pmatrix},
\end{equation}
where $\Upsilon = \operatorname{diag}(y_1,\ldots, y_n)$ and $R_2$ transforms as
\begin{equation}
\begin{aligned}
R'_2 := (QP)R_2(QP)^{-1} &= \begin{pmatrix} U^* & 0_n \\ 0_n & U \end{pmatrix}, \\
R^{\intercal \prime }_2 := (QP)R^\intercal_2(QP)^{-1} &= \begin{pmatrix} U^\intercal & 0_n \\ 0_n & U^{*\intercal} \end{pmatrix}, \\
\end{aligned}
\end{equation}
where $U$ is some general unitary matrix. Note that $*$ here denote the complex and not the Hermitian conjugate. Altogether
\begin{equation}\label{eq:twotermopt}
\begin{aligned}
\tr\lbrack Z^2 R_2 Y R_2^\intercal\rbrack &=  \tr \lbrack Z^{2 \prime} R'_{2} Y' R_2^{\intercal \prime}\rbrack 
=\tr[ \Gamma^+ U^* \Upsilon U^{\intercal}] + \tr[\Gamma^+ U \Upsilon U^{*\intercal}] 
= 2\alpha^\intercal P \beta \, ,
\end{aligned}
\end{equation}
where $\alpha(z)$ is a vector of the diagonal elements of $\Gamma^+(z)$ and $\beta$ is a vector of the diagonal elements of $\Upsilon$. $P_{ij} = |U_{ij}|^2$ making it a general unistochastic matrix, which is a subset of the bistochastic matrices \cite{Bengtsson2005}. Bistochastic matrices are those that have have non-negative entries and whose rows and columns sum to 1. It will suffice to show that, for the supremum case (at $z$ fixed),
\begin{equation}
 \sup_{X \textit{bistochastic}} \alpha^{\intercal} X \beta^{} = \alpha^{\downarrow \intercal} \beta^\downarrow.
\end{equation}
and for the infimum case that
\begin{equation}
  \inf_{X \textit{bistochastic}} \alpha^{\intercal} X \beta = \alpha^{\uparrow \intercal} \beta^\downarrow.
\end{equation}
because all permutations matrices are unistochastic \cite{Bengtsson2005}. Note that $\downarrow$ refers to rewriting the elements of the vector in descending order and the reverse for $\uparrow$. This is shown in Ref.~\cite{Coope2009} and the proof is reproduced in \ref{sec:coope}. In the supremum case we further have to take the maximum over $z_i$ which is set at $\bar{z}_i$. Similarly, the infimum taken over $z_i$ is $z_i = 1$. These two situations provide the final result stated in Eq.~(\ref{eq:prop1}).
\end{proof}

Thus we see that the maximisation of the trace term is found by setting $S = {\bar Z}=\diag({\bar z}_1, 1/{\bar z}_1,
\ldots {\bar z}_n, 1/{\bar z}_n)$ for the maximum squeezing value. The minimum is obtained by setting $S = \mathbb{I}$. Note that neither of these provide the unique maximum or minimum. We could just as well maximise by setting $S = R{\bar Z}$ for some $R \in \operatorname{OSp}(2n,\mathbb{R})$ and ${\bar Z}$ maximum. Minimisation will occur for any $S = R$. Proposition \ref{thm:tropt} shows that the instantaneous control must either squeeze or unsqueeze depending on the desired effect.

\subsection{Two-mode example}
Let us now apply our findings to the iconic case of an initial two-mode squeezed state, such as the output of a degenerate parametric down conversion process, 
or the state that would result by driving an optomechanical set-up on a blue sideband, characterised by the following covariance matrix \cite{Bucca,Ferraro2005}:
\begin{equation}
\sigma = \gamma \begin{pmatrix}
\cosh2r & 0 & \sinh2r & 0 \\
0 & \cosh2r & 0 & -\sinh2r \\
\sinh2r & 0 & \cosh2r & 0 \\
0 & -\sinh2r & 0 & \cosh2r
\end{pmatrix} \, ,
\end{equation}
with $\gamma\ge 1$. The parameter $\gamma$ allows us to range from pure states (for $\gamma=1$) 
to mixed states, with entropy growing for increasing $\gamma$.
The parameter $r$ is sometimes referred to as the `two-mode squeezing parameter' and, in a sense, 
quantifies the correlations between the two modes in this class of states. In practice, $r$ is determined 
by the strength of parametric interaction between the two modes and by the interaction time, if one assumes 
the state to be generated unitarily from the vacuum. 
In the following, we will assume the realistic value $r=0.4$.

To fix ideas, let us consider in what follows a loss rate of $100$ ${\rm KHz}$, 
and a $\chi$ parameter equal to $1.000013$, corresponding to the thermal noise 
at room temperature experienced by a mode of visible radiation at $450$ ${\rm THz}$ (a noisier situation could be construed but, as we
shall see, this typical optical circumstances will suffice to illustrate our methods).
For simplicity, we will quantify the evolution of the entropy of the relaxing states $\varrho_{t}$ through their purity 
$\mu={\rm Tr}[\varrho_t^2]$ which, for Gaussian states, turns out to be a function of the invariant $\vartheta_4[\Omega\sigma]$ alone (which is just the determinant of $\sigma$): 
$\mu=1/\sqrt{\vartheta_4[\Omega\sigma]}$ \cite{Bucca}. The steady-state purity with parameters set as above is $0.99997$.

The analytic optimisation expressed by Theorem~\ref{thm:masterm} and Proposition~\ref{thm:tropt}
may now be applied to compare the uncontrolled lossy relaxation 
of the initial states given above with the relaxation of optimally adjusted states under instantaneous, symplectic control. 

In the case of cooling, for $\gamma=2$ (such that the purity of the initial state is $0.25$), the optimal strategy 
to speed up relaxation is to undo the two-mode squeezing operation and let the vacuum state evolve through 
the channel. This allows one to attain a purity of $0.9$ after $30$ $\mu{\rm s}$, as opposed to 
the $35$ $\mu{\rm s}$ required by the uncontrolled evolution; similarly, a purity of $0.99$ is reached after 
$53$ $\mu{\rm s}$, as opposed to the $59$ $\mu{\rm s}$ of the uncontrolled evolution (clearly, the asymptotic purity is never perfectly reached, so one must resort to thresholds).

Let us now turn our attention to heating by setting $\gamma=1$ (perfectly pure initial state). Then, the fastest relaxation strategy would be to undo the two-mode squeezing of the initial state and apply maximum single-mode squeezing, as per Theorem~\ref{thm:masterm} and Proposition~\ref{thm:tropt}. In order to avoid the introduction of arbitrary maximum squeezing parameters 
(that would, in a lab, 
be dictated by specific technical constraints), let us just contrast the uncontrolled evolution of the 
initial pure two-mode squeezed state with the 
evolution of an initial vacuum state, obtained by unitarily undoing the initial two-mode squeezing: 
This will provide us with a spectacular demonstration of the effect of squeezing on relaxation dynamics. 
In mathematical terms, the difference between the two cases is captured by the trace term in Eq.~(\ref{eq:masterm}), for $n=2$ and $k=2$, which corresponds to the evolution of the determinant 
of a two-mode covariance matrix. In the case of heating, relaxation is hastened by 
an increase in the trace term: hence, whilst an initial vacuum state takes about $40$ $\mu{\rm s}$ 
to relax (to the significant digits reported above), the initial two-mode squeezed state takes as little as $4\times10^{-4}$ ${\mu}{\rm s}$! It can be shown that an instantaneous symplectic applied at this point would then be capable of stabilising not only the purity, but the whole covariance matrix (and hence the state)
to the steady-state values. Such a symplectic is the one that enacts the Williamson decomposition of the covariance matrix and can be efficiently evaluated with standard methods. 
To give a quantitative idea of the effect of such control on the fly, let us add that
the determinant of the initial two-mode squeezed state would, without it, overshoot the asymptotic determinant after $4\times10^{-4}$ ${\mu}{\rm s}$, reach a maximum 
(corresponding, in this instance, to a minimum purity of $0.855633$ at $7$ ${\mu}{\rm s}$), and then 
readjust to the steady state after a total time of about $140$ ${\mu}{\rm s}$. 


\subsection{Decoupling}\label{sec:decouple}
Proposition \ref{thm:tropt}, in conjunction with Eq.~(\ref{eq:masterm}), allows us to understand the dynamics of $\vartheta_{2k}$ when $S$ is controlled to be optimal at a given moment of time. In the optimal limit, either maximal or minimal, the value of $S$ must be either ${\bar Z}$ or $\mathbb{I}$. Using Williamson's theorem, as stated earlier, Eq.~(\ref{eq:master}) may be rewritten as
\begin{equation}
\d{(SWS^\intercal)}{t} = -SWS^\intercal + \chi \mathbb{I}.
\end{equation}
If we enact a control to set $S = Z$ or $S=\mathbb{I}$ at a given time then this equation will decouple so that it suffices to consider single mode evolution, as explored in Ref.~\cite{Carlini2014}. For the single mode case Eq.~(\ref{eq:masterm}) becomes
\begin{equation}\label{eq:prenueq}
\dot{\vartheta}_{2} = -2\vartheta_{2} + 2\chi \zeta^+_{z_i} \nu_i,
\end{equation}
where $\nu_i$ refers to the symplectic eigenvalue of the decoupled mode $i$, $z_i$ is its squeezing value obtained from the singular value decomposition and $\zeta_{z_i}^+$ is defined in Proposition \ref{thm:tropt}. Note that in the optimal maximal limit $z_i \to \overline{z}_i$ and in the minimal limit $z_i = 1$. We may talk about maximising or minimising the rate of change of $\vartheta_{2}$ by considering different limits of $z_i$. Recalling that $\vartheta_2 = \nu_i^2$ we see that Eq.~(\ref{eq:prenueq}) is equivalent to
\begin{equation}\label{eq:nueq}
\dot{\nu}_i + \nu_i - \chi\zeta_{z_i}^+ = 0,
\end{equation}
which has solution
\begin{equation}\label{eq:nusol}
  \nu_i (t)= \chi\zeta_{z_i}^+ + \left(\nu_{i0} - \chi\zeta_{z_i}^+\right)e^{-t},
\end{equation}
where $\nu_{i0}$ is the initial value of the decoupled mode $\nu_i$. The value of $\nu_i$ directly determines the entropy of the Gaussian state which for a single mode is defined as \cite{Holevo1999}
\begin{equation}
\kappa(\nu_i):= \frac{\nu_i + 1}{2} \ln \left\lbrack\frac{\nu_i+1}{2}\right\rbrack - \frac{\nu_i-1}{2} \ln\left\lbrack\frac{\nu_i-1}{2}\right\rbrack.
\end{equation}
$\kappa$ monotonically increases with $\nu_i$ and so, for a single mode, $\nu_i$ is a good entropy measure, being equal to one for pure states and greater than one for mixed states. Gaussian states have an associated temperature encoded in $\beta$, and each mode has an associated mode frequency $\omega_i$. The value of $\nu_i$ depends on the product of these two parameters via \cite{Adesso2014}
\begin{equation}
\nu_i = \frac{1+ e^{-\beta\omega_i}}{1-e^{-\beta\omega_i}}.
\end{equation}
If we fix $\omega_i$ then $\nu_i$ monotonically increases with the temperature, thus we refer to rising $\nu_i$ as heating and lowering  $\nu_i$ as cooling.

The fixed point of the evolution in Eq.~(\ref{eq:mastersol}) is at $\nu_i = \chi$. The time to reach this point diverges, hence we must fix a distance within which we are satisfied that we have arrived close enough to the target. Take this to be
\begin{equation}\label{eq:singledistance}
 |\nu_i - \chi | < \epsilon.
\end{equation}

The case $\nu_{i0} < \chi$ means that the channel is heating. The optimal strategy in this case is to increase the squeezing of each of the decoupled modes as much as possible. In other words, we choose $z_i=\bar{z}_i$ for all modes and the minimum amount of time for all to come within distance $\epsilon$ of the fixed point is
\begin{equation}\label{eq:optheat}
 T_{\textit{heat}} = \sup_{\nu_{i0}, \bar{z}_i} \ln\left[\frac{\chi\zeta_{\bar{z}_i}^+-\nu_{i0}}{\chi(\zeta_{\bar{z}_i}^+-1)+\epsilon}\right].
\end{equation}
Note that this time is in general finite and goes to zero for large values of $\bar{z}_i$ (i.e., $\zeta_{\bar{z}_i}^+$ large). 
Furthermore it is possible to set $\epsilon = 0$ and so we are able to reach exactly the fixed point in an infinitesimal amount of time for $\bar{z}_i \to \infty$ and a finite time for finite $\bar{z}_i$.

The case $\nu_{i0} > \chi$ means that the channel is cooling. In this case we optimally set $z_i = 1$ for all the decoupled modes. The minimum amount of time to come within distance $\epsilon$ of the fixed point of the dynamics is
\begin{equation}\label{eq:optcool}
 T_{\textit{cool}} =\sup_{\nu_{i0}} \ln\left\lbrack\frac{\nu_{i0} - \chi}{\epsilon}\right\rbrack.
\end{equation}
 $T_{\textit{cool}}$ diverges for $\epsilon = 0$ and so to reach the true fixed point still requires an infinite amount of time. However, in most cases it suffices to cause the state to come arbitrarily close to the thermal state of the bath. We see that in order to minimise the time to achieve this it is necessary to reduce the mode squeezing to zero, explicitly we wish to reduce the squeezing measure $\xi = \max\eig[S^\intercal S]-1$ to zero. This `unsqueezing' will need to be enacted to optimise cooling.

As it is shown in \ref{sec:onemode}, our analysis reproduces the results of Ref.~\cite{Carlini2014} for the special case of an unsqueezed bath and a single mode system. Referring to equations from Ref.~\cite{Carlini2014} with primes it is easy to show that $T_{\textit{heat}}$ coincides with $T^{\mathrm{heat}}_{\mathrm{fast}}$ (Eq.~(67')) and $T_{\textit{cool}}$ with $T^{\mathrm{cool}}_{\mathrm{fast}}$ (Eq.~(63')) provided that $\epsilon \simeq 2\sqrt{\chi^2-1}\sqrt{\epsilon'}$.

Furthermore we note that optimal control can prevent a cooling channel from ever decreasing the temperature of the state if it is sufficiently, periodically squeezed. This can be seen by noting that $\nu_i$ increases under Eq.~(\ref{eq:nueq}) when $\nu_i < \chi\zeta_{z_i}^+$ even when $\nu_i > \chi$, as seen in Fig.~\ref{fig:infinite}. 

\begin{figure}[]
  \centering
    \includegraphics[scale=0.25]{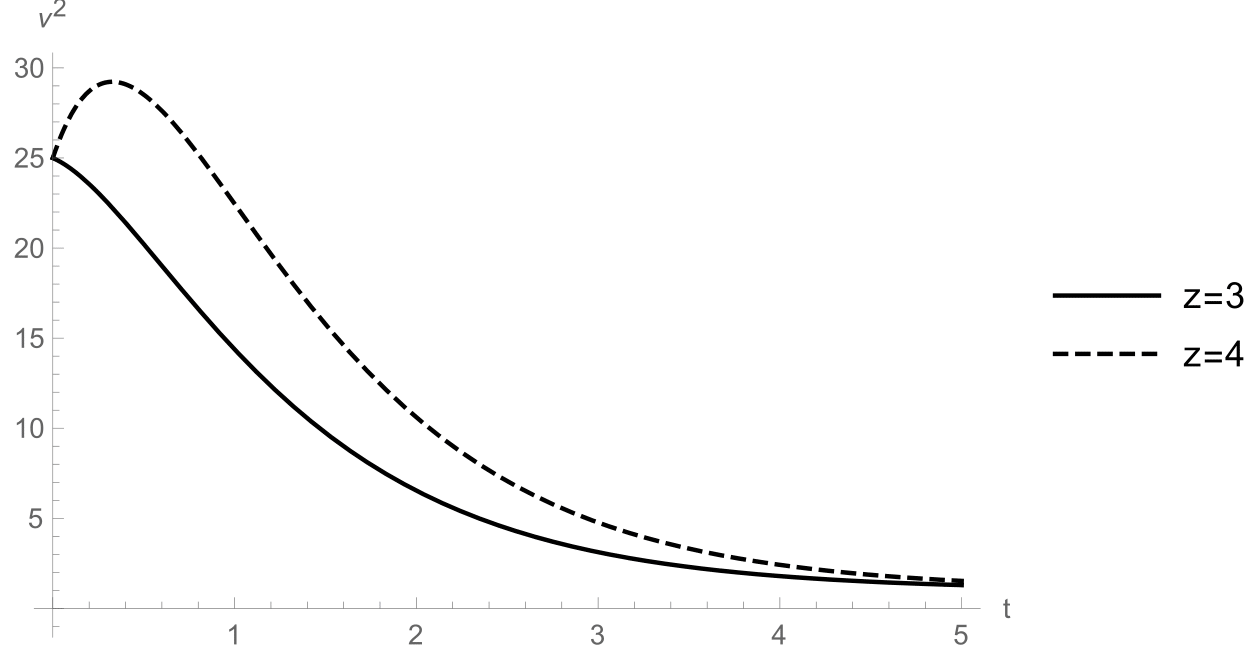}
  \caption{Showing the solution of Eq.~(\ref{eq:nusol}) for $\nu_0 = 5$, $\chi = 1$ and two different values of $z$, the initial squeezing value. $z=3$ means the entropy always decreases where as a higher level of squeezing, $z=4$ induces an initial increase of entropy.}
\label{fig:infinite}
\end{figure}

We would like to compare $T_{\textit{heat}}$ and $T_{\textit{cool}}$ with the free relaxation time of the state, in the absence of control. Computing this for a system of $n$ modes would require solving  Eq.~(\ref{eq:mastersol}) with some analogue of Eq.~(\ref{eq:singledistance}), which is not easy to achieve analytically. Ref.~\cite{Carlini2014}, however, provides results for the relaxation of single-mode Gaussian states and so a comparison can be made in this regime. 
As shown in \ref{sec:onemode}, the free decoherence time for a single mode can be written as:
\begin{equation}\label{eq:tfree}
T_{\textit{free}}= \ln \left [\frac{1}{\epsilon} \sqrt{ (\chi-\nu_0\zeta_{z_0}^+)^2+\nu_0^2 (\zeta_{z_0}^{+2}-1) \frac{(\chi^2-1)}{(\chi^2+1)} } \right ],
\end{equation}
where $z_0$ is the initial squeezing of the mode.
For the case when a single mode is involved, therefore, we can draw exactly the same conclusions as in Ref.~\cite{Carlini2014}. In particular, for vanishing tolerable errors, we get that the advantage of quantum control appears negligible if the channel is cooling while it allows for an exponential speed-up with respect to the uncontrolled dynamics for heating channels. This former statement is somehow reminiscent of the third law of thermodynamics according to which zero temperature is unattainable in finite time \cite{Masanes2016}. One can also study a measure of performance of the control procedure for the worst case scenario of possible initial conditions and obtain the same results discussed in Ref.~\cite{Carlini2014}. 

Therefore we see that squeezing and heating are intrinsically related. Squeezing itself does not affect the temperature of the state but under a lossy channel it is the key factor in increasing or decreasing its rate of change, here shown for any number of modes.

\section{Conclusion and Discussion}\label{sec:conc}
The number of degrees of freedom for an $n$-mode Gaussian state scales polynomially in $n$. Therefore when we consider the evolution of such a state it is important that we isolate the parameters that are most relevant. The information encoded in the $n$ symplectically invariant parameters allows one to extract the key property of entropy. In this paper we derived an equation for the evolution of such invariants and calculated the locally optimal rates for heating and cooling under a lossy channel. We have shown that among the time-optimal control schemes, one involves decoupling the dynamics into single modes. From here optimal cooling and heating come about by respectively unsqueezing and squeezing the decoupled modes.
As a result of the decoupling, one can extend the single-mode analysis of \cite{Carlini2014} to any number of degrees of freedom, which represents our main result. 
We have demonstrated that on the one hand heating can be accomplished arbitrarily well via optimal symplectic control in an arbitrarily short time if enough squeezing is allowed while, on the other hand, controlled cooling cannot be achieved equally well and fast as heating. 

This analysis rested on the derivation of Eq.~(\ref{eq:masterm}) but the techniques required can also provide insight into a property that is not a symplectic invariant: entanglement. The positivity of the partial transpose is necessary for the separability of a bipartite $p+q$ mode state, and sufficient if either $p$ or $q$ is equal to $1$ \cite{Lami2016,Bucca}. 
Separability is equivalent to the condition that $\tilde{\sigma} = T\sigma T$ obeys the uncertainty relation
\begin{equation}
\tilde{\sigma} + i\Omega \geq 0,
\end{equation}
where 
$T = \bigoplus_1^p \sigma_z \oplus \bigoplus_1^q \mathbb{I}_2$, where the Pauli matrix $\sigma_z =\diag(1,-1)$.
This condition is equivalent to the smallest symplectic eigenvalue of $\tilde{\sigma}$ being greater than or equal to 1 which in turn gives the separability condition \cite{Lami2016}
\begin{equation}
\tilde{\Sigma} := \sum_{k=0}^{p+q} (-1)^{p+q+k} \vartheta_{2k}[\Omega\tilde{\sigma}] \geq 0.
\end{equation}
$T$ acts by similarity and is symmetric meaning that $\tilde{\sigma}$ remains positive definite, implying that it also has a Williamson decomposition: $\tilde{\sigma} = \tilde{S}\tilde{W}\tilde{S}^\intercal$. Furthermore it obeys $\dot{\tilde{\sigma}} = T\dot{\sigma}  T = -\tilde{\sigma} + \chi \mathbb{I}$ since $T$ is time invariant. This allows us to mimic the full analysis of~\ref{sec:sympinvevol} providing
\begin{equation}\label{eq:sigmatildedot}
\begin{aligned}
\dot{\tilde{\Sigma}} = \sum_{k=1}^{n} (-1)^{n+k} \Big(-2k \vartheta_{2k}[\Omega\tilde{\sigma}] +\chi \tr[\tilde{S}\tilde{V}_{k}\tilde{S}^\intercal]\Big),
\end{aligned}
\end{equation}
where $p+q=n$ and $\tilde{V}_{k}$ has the same definition as before except for the symplectic eigenvalues of $\tilde{\sigma}$. This provides us with an evolution equation for an entanglement measure of the $p+q$ mode system, directly from the analysis employed to derive Eq.~(\ref{eq:masterm}). The difference here, however is that $\vartheta_{2k}[\Omega\tilde{\sigma}]$ is not a symplectic invariant and so further analysis will require more sophistication. 

This paper brought together techniques from linear algebra to explore the evolution of multimode Gaussian states evolving under lossy channels. As exemplified by the calculation concerning entanglement, sketched above, such techniques may be extended to provide wider analytical insight. We restricted ourselves to the specific question of finding locally optimal cooling, or heating, for states undergoing lossy channel evolution. This has been answered and shown to be intrinsically related to squeezing. We showed that the study of locally optimal trajectories can be explored using decoupled dynamics and so the single mode case provides the lower bound on the relaxation time.

\section*{Acknowledgments}

AS acknowledges financial support from EPSRC through the grant EP/K026267/1.

\appendix
\section{Evolution of symplectic invariants}\label{sec:sympinvevol}
\subsection{Recursive formulas}
Reiterating Eq.~(\ref{eq:master}) the evolution of our system obeys
\begin{equation}
\dot{\sigma} = -\sigma + \chi \mathbb{I}.
\label{eq:appmaster}
\end{equation}
Using this and the Taylor expansion in $\rm{d}t$ we obtain
\begin{equation}\label{eq:time1}
\begin{aligned}
c_{2k}[\Omega\sigma (t+{\rm d}t)] &\simeq c_{2k}[\Omega(\sigma+\dot{\sigma}{\rm d}t)] 
 =c_{2k}[\Omega \sigma - {\rm d}t \Omega \sigma + \chi \Omega {\rm d}t] 
= c_{2k}[F+G{\rm d}t],
\end{aligned}
\end{equation}
where $F:= (1-{\rm d}t)\Omega\sigma$ and $G:= \chi \Omega$.

\begin{lemma}\label{thm:coddzero}
The Taylor expansion of the following trace term to first order is
 \begin{equation}\label{eq:treven}
\begin{aligned}
&\tr[(F+G{\rm d}t)^{2k-i}] \simeq 
\begin{cases}
  \tr[F^{2k-i}] + (2k-i) {\rm d}t \tr[F^{2k-i-1}G], \, &i \text{ even}, \\
0, \, &i \text{ odd}.
\end{cases}
\end{aligned}
 \end{equation}
\end{lemma}
\begin{proof}
Expanding $(F+G{\rm d}t)^{2k-i}$ to first order in ${\rm d}t$ we obtain a single term of the form $F^{2k-i}$ and $(2k-i)$ terms of the form $F^a G{\rm d}t F^b$ where $a+b = 2k-i-1$. The cyclicity of the trace allows us to reorder these elements to obtain $\tr[F^{2k-i}] + (2k-i) {\rm d}t \tr[F^{2k-i-1}G] + o({\rm d}t)$. Now it remains to show that to first order this expression is zero for $i$ odd.

It suffices to show that 
\begin{align}
 \tr[F^{2n+1}] = 0, \quad n \in \mathbb{N}, \label{eq:zero1} \\
\tr[F^{2n}G] = 0, \quad n \in \mathbb{N}. \label{eq:zero2}
\end{align}
To prove Eq.~(\ref{eq:zero1}) we use the invariance of the trace under cyclic permutations and transposes giving
\begin{align}
\tr[(\Omega\sigma)^{2n+1}] &= \tr\big[\big((\Omega\sigma)^{2n+1}\big)^\intercal \big] 
= (-1)^{2n+1}\tr[(\sigma \Omega)^{2n+1}] 
= -\tr[(\Omega\sigma)^{2n+1}] = 0, \nonumber
\end{align}
where we used $\Omega^\intercal = -\Omega$. Eq.~(\ref{eq:zero2}) is found using a similar argument. Putting these together we prove the proposition.
\end{proof}

Recall the recursive definition for the coefficients of the characteristic function given in Sec.~\ref{sec:sympinv}:
\begin{equation}\label{eq:ck}
c_k[X] = \frac{-1}{k} \sum_{i=0}^{k-1} \tr [X^{k-i}] c_i[X].
\end{equation}
Recalling that $c_{2k}[\cdot] \equiv \vartheta_{2k}[\cdot]$ we change notation to start considering these symmetric functions. Using the recursive formula and Lemma \ref{thm:coddzero} we see that
\begin{equation}\label{eq:levext3}
\vartheta_{2k}[F+G{\rm d}t] = \frac{-1}{2k} \sum_{i=0}^{k-1} \tr [(F+G{\rm d}t)^{2(k-i)}] \vartheta_{2i}[F+G{\rm d}t].
\end{equation}

\begin{lemma}\label{thm:cexpansion}
Taylor expanding $\vartheta_{2k}[F+G{\rm d}t]$ we arrive at
\begin{equation}\label{eq:cexp}
\begin{aligned}
\vartheta_{2k}&[F+G{\rm d}t] = 
  \vartheta_{2k}[F]- {\rm d}t \sum_{i=0}^{k-1}\tr[(\Omega\sigma)^{2(k-i)-1}G]\vartheta_{2i}[F] +o({\rm d}t).
\end{aligned}
\end{equation}
\end{lemma}
\begin{proof}
From Eq.~(\ref{eq:levext3}) we can explicitly show that
 \begin{equation}\label{eq:early}
\begin{aligned}
\vartheta_0[F+G{\rm d}t] &= 1, \quad 
\vartheta_2[F+G{\rm d}t] &= \vartheta_2[F] - {\rm d}t \tr[\gamma G] \, ,
\end{aligned}
\end{equation}
where we define $\gamma :=\Omega\sigma$ for brevity in the proceeding proof and where we have also used Lemma \ref{thm:coddzero}.
From here we will proceed with an inductive proof.  We assume that Eq.~(\ref{eq:cexp}) holds for some $k$ and then show that if this is true then it holds for $k+1$. 

Using Eq.~(\ref{eq:levext3}) we may write the expansion out for $k+1$ and then use Lemma \ref{thm:coddzero} and the definition of $F$ to arrive at
\begin{equation}\label{eq:2kp1one}
\begin{aligned}
\vartheta_{2(k+1)}[F+G{\rm d}t] =& 
 \frac{-1}{2(k+1)} \sum_{i=0}^{k}\Big( \left(1-2(k+1-i){\rm d}t\right)\tr[\gamma^{2(k+1-i)}] \\
&+ 2(k+1-i) {\rm d}t \tr[\gamma^{2(k+1-i)-1}G] \Big) \vartheta_{2i}[F+G{\rm d}t].
\end{aligned}
\end{equation}
Substituting Eq.~(\ref{eq:cexp}) into Eq.~(\ref{eq:2kp1one}) we extract the first two terms that look like they would fulfill the proof plus a final one that we would hence like to show is zero:
\begin{equation}
\begin{aligned}
\vartheta_{2(k+1)}[&F+G{\rm d}t] = \vartheta_{2(k+1)}[F] -{\rm d}t \sum_{i=0}^k \tr[\gamma^{2(k+1-i)-1}G] \vartheta_{2i}[F] \\
&+ \frac{\rm{d}t}{2(k+1)} \sum_{i=1}^{k}\Bigg( \tr[\gamma^{2(k+1-i)}]\sum_{j=0}^{i-1} \tr[\gamma^{2(i-j)-1}G] \vartheta_{2j}[F]+ 2i\tr[\gamma^{2(k+1-i)-1}G] \vartheta_{2i}[F]\Bigg).
\end{aligned}
\end{equation}
Dropping $(2(k+1))^{-1}$ we proceed to examine the final piece, referring to it as $L$, and rewriting it as
\begin{equation}\label{eq:extra}
\begin{aligned}
L = &\rm{d}t\biggl \{\sum_{i=1}^{k-1}\tr[\gamma^{2(k-i)}]\sum_{j=0}^{i-1} \tr[\gamma^{2(i-j)-1}G] \vartheta_{2j}[F]
+ \sum_{i=1}^{k-1}2i\tr[\gamma^{2(k-i)-1}G] \vartheta_{2i}[F] \biggr \}.
\end{aligned}
\end{equation}
Note that we have relabelled $k$ as $k-1$ to shorten the expression but it will not alter the analysis. The ${\rm d}t$ at the front reminds us that everything should be expanded to zeroeth order inside the sum. To prove the lemma it is necessary to show that $L\equiv 0$.

Expanding $\vartheta_{2i}[F]$ to introduce another sum we arrive at
\begin{equation}\label{eq:nearly}
\begin{aligned}
L = & {\rm d}t\sum_{i=1}^{k-1}\sum_{j=0}^{i-1}\Bigg( \tr[\gamma^{2(k-i)}] \tr[\gamma^{2(i-j)-1}G]
- \tr[\gamma^{2(k-i)-1}G] \tr[\gamma^{2(i-j)}] \Bigg) \vartheta_{2j}[F].	
\end{aligned}
\end{equation}
From here note that for a general sum with elements $Y_{ij}$ we have
\begin{equation}\label{eq:summation}
\begin{aligned}
 \sum_{i=1}^{k-1} \sum_{j=0}^{i-1} Y_{ij} = \sum_{j=0}^{k-2} \sum_{i=j+1}^{k-1} Y_{ij} 
 =\frac{1}{2} \sum_{j=0}^{k-2} \sum_{i'=0}^{k-j-2} Y_{i'+j+1,j} + \frac{1}{2}\sum_{j=0}^{k-2}\sum_{i''=0}^{k-j-2} Y_{k-i''-1,j},
 \end{aligned}
\end{equation}
where $i' = i-(j+1)$ and $i'' = k-i'-j-2$. The first equality of Eq.~(\ref{eq:summation}) can be seen with observation. The second involves a redefinition of the sums where we split them into two halves and then redefine the labels such that one is descending whilst the other ascends. When such a summation redefinition is applied to Eq.~(\ref{eq:nearly}) it will be clear that $L \equiv 0$.

Thus we prove that if $\vartheta_{2k}[F+G{\rm d}t]$ is given in Eq.~(\ref{eq:cexp}) then this also holds for $k\to k+1$. From Eq.~(\ref{eq:early}) we see that it is true for $k=1$ and so, inductively it is true for all $k$. To write it in the form stated one must replace $\gamma$ with $\Omega\sigma$.
\end{proof}

\begin{lemma}\label{thm:ca}
Taylor expanding $\vartheta_{2k}[F]$ using $F := (1-{\rm d} t) \Omega \sigma$ we find that
\begin{equation}
 \vartheta_{2k}[F] = (1-2k{\rm d}t)\vartheta_{2k}[\Omega\sigma].
\end{equation}
\end{lemma}
\begin{proof}
Expanding out the recursive formula and again defining $\gamma := \Omega \sigma$ we get a product of sums of the form
\begin{equation}\label{eq:crossing}
\begin{aligned}
\vartheta_{2k}[F] =& \Bigg[\frac{-1}{2k} \sum_{i_1=0}^{k-1} \Big(1-2(k-i_1){\rm d}t\Big)\Bigg] 
\Bigg[\frac{-1}{2i_1} \sum_{i_2=0}^{i_1-1} \Big(1-2(i_1-i_2){\rm d}t\Big) \Bigg] 
\ldots \\
&\ldots \Bigg[\frac{-1}{2i_{k-1}} \sum_{i_k=0}^{i_{k-1}-1} \Big(1-2(i_{k-1}-i_k){\rm d}t\Big) \Bigg]  
 \tr[\gamma^{2(k-i)}] \ldots \tr[\gamma^{2(i_{k-1}-i_k)}].
\end{aligned}
\end{equation}
By only keeping terms that are less than second order in ${\rm d}t$ we get a smaller sum
\begin{equation}
(1-2k)\vartheta_{2k}[\Omega \sigma] + X
\end{equation}
$X$ consists of the remaining terms which come in pairs. Take for instance the first pair which is generated by choosing the $+2i_1{\rm d}t$ coefficient in the first line of Eq.~(\ref{eq:crossing}), with everything else at zeroeth order, and secondly the $-2i_1{\rm d}t$ coefficient in the second line, with everything else at zeroeth order. The pairs will each cancel to become zero. The final piece comes without a partner but has coefficient $i_k = 0$, and so does not contribute. Therefore $X\equiv 0$ and the lemma is proven.
\end{proof}

\begin{lemma}\label{thm:mast}
Using the previous two Taylor expansions we may write the rate of change of $\vartheta_{2k}$ as
\begin{equation}\label{eq:master1app}
\begin{aligned}
\dot{\vartheta}_{2k}[\Omega\sigma]=&-2k \vartheta_{2k}[\Omega\sigma]  
- \chi\sum_{i=0}^{k-1}\tr[(\Omega\sigma)^{2(k-i)-1}\Omega]\vartheta_{2i}[\Omega\sigma].
\end{aligned}
\end{equation}
\end{lemma}
\begin{proof}
The first Taylor expansion came from Lemma \ref{thm:cexpansion} stating that
\begin{equation}
\begin{aligned}
 \vartheta_{2k}&[F+G{\rm d}t] \simeq 
 \vartheta_{2k}[F] - {\rm d}t \sum_{i=0}^{k-1}\tr[F^{2(k-i)-1}G]\vartheta_{2i}[F].\nonumber
\end{aligned}
\end{equation}
Using Lemma \ref{thm:ca} and the same reasoning as in Eq.~(\ref{eq:time1}) we rewrite the above as
\begin{equation}\label{eq:masterminc}
\begin{aligned}
 \vartheta_{2k}[\Omega\sigma(t+{\rm d}t)]-\vartheta_{2k}[\Omega\sigma] \simeq 
 -2k{\rm d}t \vartheta_{2k}[\Omega\sigma] - {\rm d}t \sum_{i=0}^{k-1}\tr[(\Omega\sigma)^{2(k-i)-1}G]\vartheta_{2i}[\Omega\sigma].\nonumber
\end{aligned}
\end{equation}
Dividing through by ${\rm d}t$ we prove the proposition, recalling that $G:= \chi \Omega$.
\end{proof}

\subsection{Telescoping the series}
Eq.~(\ref{eq:master1app}) provides a differential equation describing the rate of change of $\vartheta_{2k}$. We now work towards a simplification of this equation using Williamson's theorem \cite{Williamson1936, Olivares2011} and noticing that the series `telescopes' to provide a simpler form.

Given Williamson decomposition, $\sigma = SWS^\intercal$, as described earlier, and the symplectic property: $S^\intercal \Omega S = \Omega$, we may rewrite the trace term that appears in the sum of Eq.~(\ref{eq:master1app}):
\begin{equation}
\begin{aligned}
 \tr[(\Omega\sigma)^{2k-1} \Omega] &= \tr[\overbrace{\Omega SWS^\intercal \ldots \Omega SWS^\intercal}^{2k-1} \Omega] 
= -\tr[S\overbrace{WS^\intercal \Omega S \ldots WS^\intercal \Omega S}^{2k-2} WS^\intercal ] \\
&= -\tr[S \overbrace{W\Omega \ldots W\Omega}^{2k-2} W S^\intercal]
= -\tr[S W^{2k-2} \Omega^{2k-2} WS^\intercal] \\
&= (-1)^{k} \tr[SW^{2k-1} S^\intercal].\nonumber
\end{aligned}
\end{equation}
where we used $(W\Omega)^2=W^2\Omega^2$, $\Omega^2 = -\mathbb{I}$ as well as cyclic properties of the trace. Thus Eq.~(\ref{eq:master1app}) becomes
\begin{equation}\label{eq:master2app}
\begin{aligned}
\dot{\vartheta}_{2k}[\Omega\sigma]=&-2k \vartheta_{2k}[\Omega\sigma]  
-\chi \sum_{i=0}^{k-1}(-1)^{k-i} \tr[SW^{2(k-i)-1}S^\intercal]\vartheta_{2i}[\Omega\sigma].
\end{aligned}
\end{equation}
Using definitions given in Sec.~\ref{sec:sympinv} we may now prove the theorem stated in the main body:
\begin{thm-hand}[2]
For $\sigma$ evolving under $\dot{\sigma} = -\sigma + \chi\mathbb{I}$, the evolution of the symplectic invariants, defined by $\vartheta_{2k}[\Omega\sigma]$, obeys
\begin{equation}\label{eq:mastermapp}
\dot{\vartheta}_{2k}[\Omega\sigma] = -2k\vartheta_{2k}[\Omega\sigma] +\chi \tr[SV_{k}S^\intercal],
\end{equation}
where
\begin{equation}
V_{k} = \bigoplus_{i=1}^n \left(\nu_i \bar{\vartheta}_{2(k-1)}^i[\Omega\sigma]\right)\mathbb{I}_2.
\end{equation}
\end{thm-hand}
\begin{proof}
Beginning with the final term of Eq.~(\ref{eq:master2app}) and dropping the $\chi$ factor we may take the sum inside the trace to give
 \begin{equation}\label{eq:telescope}
 \begin{aligned}
\sum_{i=0}^{k-1} (-1)^{k+1}(-1)^i \tr[SW^{2(k-i)-1}S^\intercal]\vartheta_{2i} =
  \tr\left\lbrack S \sum_{i=0}^{k-1} (-1)^{k+1} (-1)^i W^{2(k-i)-1} \vartheta_{2i} S^\intercal\right\rbrack,
  \end{aligned}
 \end{equation}
which allows us to consider the internal sum before tracing. Recalling that the diagonal elements of $W$ occur in pairs we only need to consider $n$ symplectic eigenvalues denoted $\nu_q$, $q = 1,\ldots,n$. The $i$th term of the sum in Eq.~(\ref{eq:telescope}) takes the form
\begin{equation}\label{eq:tele1}
\begin{aligned}
(-1)^{k+1}(-1)^i \Bigg(\nu_q^{2(k-i)+1} \bar{\vartheta}^q_{2(i-1)} + \nu_q^{2(k-i)-1}\bar{\vartheta}_{2i}^q\Bigg),
\end{aligned}
\end{equation}
and the $i+1$th takes the form
\begin{equation}\label{eq:tele2}
 (-1)^{k+1}(-1)^{i+1} \Bigg(\nu_q^{2(k-i)-1} \bar{\vartheta}^q_{2i} + \nu_q^{2(k-i)-3}\bar{\vartheta}_{2(i+1)}^q\Bigg).
\end{equation}
Recall that $\bar{\vartheta}_a^b$ is a sort of reduced symplectic invariant which is similar to $\vartheta_a$ but where we remove all terms involving the $b$th symplectic eigenvalue. By splitting up Eq.~(\ref{eq:telescope}) to consider each $\nu_q$ individually and also by spltting each $i$th term we see that telescoping is going to occur. This is the situation when the latter piece of the $i$th term cancels the former term of the $i+1$th term. Noting that this cancellation is going to occur between Eq.~(\ref{eq:tele1}) and Eq.~(\ref{eq:tele2}) we should be left with the very first and last pieces of the entire series. The first piece is equal to zero and so we are just left with the final term
\begin{equation}
 (-1)^{k+1}(-1)^{k-1} \nu_q \bar{\vartheta}^q_{2(k-1)} = \nu_q \bar{\vartheta}^q_{2(k-1)},
\end{equation}
and therefore the expressions of Eq.~(\ref{eq:telescope}) are equal to $\tr\left\lbrack S V_{k} S^\intercal\right\rbrack$ where 
\begin{equation}
V_{k} = \bigoplus_{i=1}^n \left(\nu_i \bar{\vartheta}_{2(k-1)}^i[\Omega\sigma]\right)\mathbb{I}_2.
\end{equation}
\end{proof}

\section{Bistochastic limits}\label{sec:coope}
Here, we prove an optimisation result regarding bistochastic matrices. We begin with the reiteration of the notation that, for some vector $v$, $v^\downarrow$ denotes a new vector of elements of $v$ written in descending order. $v^\uparrow$ is similar but for ascending order.
\begin{lemma}
Let $\alpha$ and $\beta$ be two real vectors of length $m$. The supremum of the following inner product is given by
\begin{equation}
 \sup_{X \textit{bistochastic}} \alpha^{\intercal} X \beta^{} = \alpha^{\downarrow \intercal} \beta^\downarrow.
\end{equation}
and the infimum by
\begin{equation}
  \inf_{X \textit{bistochastic}} \alpha^{\intercal} X \beta = \alpha^{\uparrow \intercal} \beta^\downarrow.
\end{equation}
\end{lemma}
\begin{proof}
Taken from Ref.~\cite{Coope2009}. We begin with a consideration of the supremum and then turn our attention to the infimum. The permutation matrices form a subset of the bistochastic matrices and so we may alter $X$ to make $\alpha$ and $\beta$ descending. Then we fix $X$ to some specific bistochastic matrix to define
\begin{equation}
\chi := \alpha^{\downarrow \intercal} X \beta^{\downarrow} \equiv \sum_{i,j=1}^m a_i b_j X_{ij},
\end{equation}
where $a_i$ and $b_i$ are the elements of $\alpha^\downarrow$ and $\beta^\downarrow$ respectively. Consider $X \neq \mathbb{I}$. Let $k$ be the smallest index $i$ such that $X_{ii} \neq 1$. Note that for $i < k$, $X_{ii} = 1$ and therefore $X_{ij} = 0$ if $i< k$ and $i \neq j$, or if $j<k$ and $i\neq j$. Since $X_{kk} < 1$, then for some $l > k$, $X_{kl} > 0$. Likewise, for some $p>k$, $X_{pk} > 0$. These imply that $X_{pl} \neq 1$.

The inequalities above mean that we can choose $\epsilon >0$ such that the matrix $X'$ is bistochastic where
\begin{align}
X'_{kk} &= X_{kk} + \epsilon, \nonumber \\
X'_{kl} &= X_{kl} - \epsilon, \nonumber\\
X'_{pk} &= X_{pk}- \epsilon, \nonumber\\
X'_{pl} &= X_{pl} + \epsilon, \nonumber
\end{align}
and $X'_{ij} = X_{ij}$ in all other cases. Now define
\begin{equation}
\chi' = \sum_{i,j = 1} a_i b_j X'_{ij}.
\end{equation}
Recalling that $l>k$ and $p>k$, so that $a_k<a_p$ and $b_k<b_l$,
\begin{equation}
\begin{aligned}
\chi' - \chi = \epsilon(a_k b_k - a_kb_l - a_mb_k + a_mb_l) 
=\epsilon(a_k -a_m)(b_k-b_l) 
\geq 0,
\end{aligned}
\end{equation}
which means that the term $\sum a_ib_j X_{ij}$ is not decreased. $\epsilon$ may be chosen to reduce an off-diagonal term in $X$ to zero without affecting the bistochasticity of $X$ and without decreasing the value of $\chi'-\chi$. As this process is iterated $X$ may be brought to identity $\mathbb{I}$ without decreasing $\chi$, achieving its maximum value.

The argument for the infimum is completely analogous except we begin by considering
\begin{equation}
\chi := \alpha^{\uparrow \intercal} X \beta^{\downarrow} \equiv \sum_{i,j=1}^m c_i b_j X_{ij},
\end{equation}
where $c_i$ are the elements of $\alpha^\uparrow$. The analysis is similar except that we will be considering $\chi' - \chi \leq 0$, and so transforming $X$ into $\mathbb{I}$ does not \textit{increase} the value of $\chi$. Thus we find $X = \mathbb{I}$ again but for the situation where the vectors are oppositely ordered.
\end{proof}

\section{The case of a single mode}\label{sec:onemode}

Here we compare $T_{\textit{heat}}$ and $T_{\textit{cool}}$ with $T^{\mathrm{heat}}_{\mathrm{fast}}$, 
$T^{\mathrm{cool}}_{\mathrm{fast}}$ and $T_{\mathrm{free}}$ derived in Ref.~\cite{Carlini2014} for the case of a single mode
coupled to an unsqueezed bath.
For clarity, equations quoted from Ref.~\cite{Carlini2014} will be primed.
Both papers begin with Eq.~(\ref{eq:covmastergen}) which allows us to identify
\begin{equation}
\gamma \leftrightarrow \eta, \quad M_1=M_2 \leftrightarrow 0, \quad \mu \leftrightarrow 1/\nu,
\end{equation}
with symbols from Ref.~\cite{Carlini2014} on the left and ours on the right. 
Eqs.~(38')-(39') show the fixed point of the dynamics to be
\begin{equation}
r_{\mathrm{fp}}= 0, \mu_{\mathrm{fp}}\leftrightarrow 1/\chi.
\end{equation} 
Using Eq.~(33') and noting that $\mathrm{Tr} (\sigma) =2\nu \zeta_z^+$ we identify
\begin{equation}
\cosh 2r  \leftrightarrow \zeta_z^+
\end{equation}
and maximal squeezing in $\cosh 2r_M \leftrightarrow\zeta_{\bar z}^+$, while for the initial conditions we identify
\begin{equation}
\mu_0\leftrightarrow 1/\nu_0, \quad \cosh 2r_0\leftrightarrow \zeta_{z_{\textit{0}}}^+,
\end{equation}
where $z_{\textit{0}}$ is the initial squeezing value. Ref.~\cite{Carlini2014} uses the fidelity Eqs.~(50')-(51') for the tolerable error while we use Eq.~(\ref{eq:singledistance}). To avoid confusion, we rename $\epsilon '$ the allowed error in Eq.~(50'). In particular, we can rewrite Eqs.~(50') and (66') for the heating and cooling cases, respectively, as 
$\mu_{\mathrm{Th/Tc}}\leftrightarrow \frac{1}{\chi}\left [1(+/-) 2\frac{\sqrt{\epsilon '}}{\chi}\sqrt{\chi^2-1}\right ]$. 
With all of the above taken into account, it is then very easy to verify that Eq.~(67') and Eq.~(63') coincide with our optimal times for decoupled modes, i.e. that 
\begin{equation}
T^{\mathrm{heat}}_{\mathrm{fast}}\leftrightarrow T_{\textit{heat}}, \quad T^{\mathrm{cool}}_{\mathrm{fast}}\leftrightarrow T_{\textit{cool}}
\end{equation}
provided that $\epsilon \simeq 2\sqrt{\chi^2-1}\sqrt{\epsilon '}$.
Finally, the free decoherence time Eq.~(54') is translated as Eq.~(\ref{eq:tfree}).

\section*{References}
\bibliographystyle{iopart-num}
\bibliography{mybib}{}

\end{document}